\documentclass[letterpaper,twocolumn,showpacs]{revtex4}

\pdfoutput=1

\usepackage{amsmath}
\usepackage{amsfonts}
\usepackage{amsthm} 
\usepackage{graphicx}

\usepackage{hyperref}
\hypersetup{%
pdftitle  = {Spectrum conditions for symmetric extendible states},
pdfauthor = {Geir Ove Myhr and Norbert L{\374}tkenhaus},
pdfkeywords = {quantum state, symmetric extension, extendible, extendable, extensible, extendibility, extendability, extensibility},
colorlinks,
linkcolor=blue,
filecolor=green,
urlcolor=blue,
citecolor=blue
}

\DeclareMathOperator{\tr}{tr}
\DeclareMathOperator{\rank}{rank}

\DeclareMathOperator{\mathspan}{span}

\newcommand{\ie}{i.~e.~}
\newcommand{\spec}[1]{\vec{\lambda}(#1)}
\newcommand{\swap}{P_{BB^\prime}}
\newcommand{\swapBE}{P_{BE}}
\newcommand{\ket}[1]{|#1\rangle}
\newcommand{\bra}[1]{\langle#1|}
\newcommand{\proj}[1]{\ket{#1}\bra{#1}}
\newcommand{\ketbra}[2]{\ket{#1}\bra{#2}}
\newcommand{\braket}[2]{\langle #1 | #2 \rangle}
\newcommand{\im}{\mathrm{i}}
\newcommand{\e}{\mathrm{e}}

\begin{document}

\title{Spectrum conditions for symmetric extendible states}
\author{Geir Ove \surname{Myhr}}
\email{gomyhr@iqc.ca}
\author{Norbert L{\"u}tkenhaus}

\affiliation{
Institut f\"{u}r Theoretische Physik I, \& Max-Planck Research Group, Institute of Optics, Information and Photonics, 
Universit\"{a}t Erlangen-N\"{u}rnberg, Staudtstra{\ss}e 7/B2, 91058 Erlangen, Germany\\}

\affiliation{Institute for Quantum Computing \& Department of Physics and Astronomy, University of Waterloo, University Ave.~W.~N2L 3G1, Canada}

\date{\today}

\begin{abstract}
We analyze bipartite quantum states that admit a symmetric extension. Any such state can be decomposed into a convex combination of states that allow a \emph{pure} symmetric extension. A necessary condition for a state to admit a pure symmetric extension is that the spectra of the local and global density matrices are equal. This condition is also sufficient for two qubits, but not for any larger systems. Using this condition we present a conjectured necessary and sufficient condition for a two qubit state to admit symmetric extension, which we prove in some special cases. The results from symmetric extension carry over to degradable and anti-degradable channels and we use this to prove that all degradable channels with qubit output have a qubit environment.
\end{abstract}


\maketitle

\newtheorem{theorem}{Theorem}
\newtheorem{lemma}{Lemma}
\newtheorem{corollary}{Corollary}
\newtheorem{proposition}{Proposition}

\theoremstyle{definition}
\newtheorem{conjecture}{Conjecture}
\newtheorem{example}{Example}
\newtheorem{remark}{Remark}


\section{Introduction}


Different bipartite quantum states can be useful for different tasks, and one of the goals of quantum information theory is to find out which properties are required from a state for it to be a useful resource for a given task. Some mathematical properties of the states can tell something about what you can or cannot do with it. For example, if the partial transpose of a state is a positive semidefinite operator it is not possible to distill entanglement from that state, no matter how many copies one has available \cite{horodecki98a}. Similarly, for a state $\rho_{AB}$, if the operator $I_A \otimes \rho_B - \rho_{AB}$ is not positive semidefinite, it \emph{is} possible to distill entanglement from many copies \cite{horodecki99a}. For distilling secret key, the only known precondition is that the state must be entangled \cite{curty04a}, \ie it is not possible to express it as a convex combination of pure product states.


One can consider the tasks of distilling entanglement or secret key using classical communication in one direction. In this work we will consider communication from a party named Alice in possesion of system $A$ to a party named Bob in possesion of system $B$. 
If a state admits a \emph{symmetric extension} to two copies of $B$ none of these tasks will be possible, 
due to the monogamy of entanglement and secret key. The focus in this work is on characterizing the states that admit a symmetric extension.

The bipartite quantum states we consider live on the system $AB$ with the two subsystems $A$ and $B$. The corresponding Hilbert spaces are $\mathcal{H}_{A}$, $\mathcal{H}_{B}$ and $\mathcal{H}_{AB} = \mathcal{H}_{A} \otimes \mathcal{H}_{B}$. We want to extend the system to a system $B^\prime$ which is a copy of $B$, and with an isometry between the two, so that for an operator on or vector in $\mathcal{H}_B$, there is a corresponding one in $\mathcal{H}_{B^\prime}$. The extended system is $ABB^\prime$ with Hilbert space $\mathcal{H}_{ABB^\prime} = \mathcal{H}_{A} \otimes \mathcal{H}_{B} \otimes \mathcal{H}_{B^\prime}$.

Because of the isometry, we can define the swap operator $\swap$ as the unitary operator that interchanges states on the two systems $B$ and $B^\prime$. 
In terms of corresponding orthogonal bases $\swap = \sum_{ij} \ket{ij}\bra{ji}$. 
The swap is a Hermitian operator, since it is unitary and $\swap^2 = I$. 
We say that a state $\rho_{ABB^\prime}$ is \emph{symmetric} if $\rho_{ABB^\prime} = \swap \rho_{ABB^\prime} \swap^\dagger$. 
For the main part of this paper we ignore whether a state has support on the symmetric subspace (states that satisfy $1/2(I+\swap)\rho_{ABB^\prime}1/2(I+\swap)^\dagger$), antisymmetric subspace (states that satisfy $1/2(I-\swap)\rho_{ABB^\prime}1/2(I-\swap)^\dagger$) or both, and in general it will be a mixture between the two (but see appendix \ref{app:bosonicfermionic}).

Finally we say that a bipartite state $\rho_{AB}$ has a \emph{symmetric extension} (or is \emph{symmetric extendible}) if there exists a tripartite state $\sigma_{ABB^\prime}$ such that $\tr_{B^\prime} \sigma_{ABB^\prime} = \rho_{AB}$ and $\sigma_{ABB^\prime} = \swap \sigma_{ABB^\prime} \swap^\dagger$, \ie $\sigma_{ABB^\prime}$ is symmetric.


In general one can consider extensions to $n_A$ copies of system $A$ and $n_B$ copies of system $B$ and this is called a $(n_A,n_B)$-symmetric extension. 
This has been used to derive algorithms for deciding whether a state is entangled or separable \cite{doherty02a}. 
Questions like whether a state admits symmetric extensions can also be formulated as quantum marginal problems \cite{bravyi04a,han05a,butterley06a,liu06a,hall07a}. 
Asking if a state $\rho_{AB}$ has a $(1,N)$-symmetric extension is just a special case of the marginal problem of deciding if there exists a state on the $N+1$ systems $A, B_1, \ldots, B_N$ with given reduced states $\rho_{AB_j}$. 
This becomes a symmetric extension when one demands that all $\rho_{AB_j}$ are equal to the given $\rho_{AB}$ which is to be extended. 
If one such state exists it can always be symmetrized to give a state that is invariant under any permutations of the $B_j$.

Since we are interested in the one-way communication aspect we will only be considering $(1,2)$-symmetric extensions. 
In this setting any state $\rho_{AB}$ where $\rho_A := \tr_B\rho_{AB}$ is maximally mixed corresponds to a channel through the Choi-Jamio{\l}kowski isomorphism, and those that are also symmetric extendible correspond to \emph{anti-degradable} channels \cite{wolf07a}.


The reason symmetric extension is interesting in a one-way classical communications setting, is that no matter what operations Alice and Bob perform, the state will keep a symmetric extension if communication from Bob to Alice is not allowed. 
\begin{lemma}
 (Nowakowski and Horodecki \cite{nowakowski07suba}) Let $\Lambda$ be a (not necessarily trace-preserving) quantum operation that can be realized with 1-LOCC, \ie it is of the form 
\begin{equation}
  \Lambda(\rho) = \sum_{ij}(I \otimes B_{ij})(A_i \otimes I) \rho (A_i \otimes I)^\dagger (I \otimes B_{ij})^\dagger
\end{equation}
where $\sum_{i} A_i^\dagger A_i \le I$ and $\sum_j B_{ij}^\dagger B_{ij} = I$ for all $i$ since Bob cannot communicate the outcome of a probabilistic operation back to Alice. 

If $\rho_{AB}$ admits a symmetric extension of the system $B$, then so does $\Lambda(\rho_{AB})$.
\end{lemma}

An interesting special case is when Alice performs an invertible filter operation and Bob performs a unitary. Then the operation can be reversed with non-zero probability, so the output state admits a symmetric extension if \emph{and only if} the input state admits one. 

Knowing when a state admits a symmetric extension can also be useful in the analysis of two-way distillation protocols for entanglement or secret key. A two-way protocol consists of a finite number of one-way rounds going in alternating directions. Before the last round, the state cannot have a symmetric extension to two copies of the receiving party's system if the protocol is to succeed \cite{myhr08suba}. 

This paper is organized as follows. In section \ref{sec:pure-extendible-decomp} we show that any state with symmetric extension can be written as a convex combination of states with pure symmetric extension. In section \ref{sec:condition} we give a necessary condition for a state to have a pure symmetric extension. This condition is proved to be sufficient for two qubits in section \ref{sec:necsufqubit} and section \ref{sec:counterex} shows that this is not true for any higher dimension. In section \ref{sec:2qubits} we give a conjectured necessary and sufficient condition for a 2-qubit state, which we prove in some special cases. The techniques from the previous sections are applied to anti-degradable and degradable channels in section \ref{sec:degradable}.


\section{Decomposition into pure-symmetric extendible states}
\label{sec:pure-extendible-decomp}

Separable quantum states are those states that can be written as convex combinations of product states $\rho_A \otimes \rho_B$ and they can even be decomposed further into convex combinations of pure product states. I.~e.
\begin{equation}
  \rho_\text{sep} = \sum_j p_j \proj{\psi_j} \otimes \proj{\phi_j}.
\end{equation}
Although it can be difficult to determine whether or not a given state can be written on this form or not --- and if it can, to find some $\ket{\psi_j}$ and $\ket{\phi_j}$ explicitly --- the fact that all separable states can be written like this allows us to prove properties of separable states in general. 

One may ask if there is an analog to this for states that allow for a symmetric extension. Clearly, it is not true that any $\rho_{AB}$ that allows for symmetric extension can be decomposed into pure states with the same property. This is because the only pure states that allow for symmetric extension are the product states, and their convex hull is the set of separable states. But it turns out that if we consider the \emph{extended} states --- the $\rho_{ABB^\prime}$ that are invariant under exchange of $B$ and $B^\prime$ --- they can be written as convex combinations of pure states with the same property. In fact, the pure states in the spectral decomposition can be chosen to have this property.

\begin{lemma}
\label{lem:symspecdecomp}
A tripartite state $\rho_{ABB^\prime}$ which is invariant under exchange of $B$ and $B^\prime$, $\rho_{ABB^\prime} = P_{BB^\prime} \rho_{ABB^\prime} P_{BB^\prime}^\dagger$, can be written in the spectral decomposition
\begin{equation}
  \rho_{ABB^\prime} = \sum_j \lambda_j \proj{\phi_j}
\end{equation} 
in such a way that $\proj{\phi_j} = P_{BB^\prime} \proj{\phi_j} P_{BB^\prime}^\dagger$, \ie $P_{BB^\prime} \ket{\phi_j} = \pm \ket{\phi_j}$.
\end{lemma}
\begin{proof}
  Since $\rho_{ABB^\prime} = \swap \rho_{ABB^\prime} \swap^\dagger$, $\rho_{ABB^\prime} \swap = \swap \rho_{ABB^\prime}$, so $\rho_{ABB^\prime}$ and $\swap$ are commuting diagonalizable operators and therefore have a common set of eigenvectors. Since $\swap^2 = I$, $\swap$ has eigenvalues $\pm 1$ and all its eigenvectors therefore satisfy $\swap \ket{\phi} = \pm \ket{\phi}$.
\end{proof}

The above lemma applies to the extended state $\rho_{ABB^\prime}$, but our main interest is for states $\rho_{AB}$ that \emph{admit} a symmetric extension. By tracing out the $B^\prime$ system we get
\begin{corollary}
\label{cor:pure-extendible-decomp}
  A bipartite quantum state $\rho_{AB}$ admits a symmetric extension if and only if it can be written as a convex combination 
  \begin{equation}
  \rho_{AB} = \sum_j p_j \rho_{AB}^j; \quad 0 \le p_j \le 1; \quad \sum_j p_j = 1
  \end{equation}
  of states $\rho_{AB}^j$ which allow a \emph{pure} symmetric extension.
\end{corollary}

Hence, all the extremal states in the convex set of symmetric extendible states are extendible to pure states. We will call those states \emph{pure-extendible}. In the next section 
we give a simple necessary condition for a state to be pure-extendible, and in the following sections we show that it is sufficient if and only if it is a state on two qubits.


\section{The spectrum condition}
\label{sec:condition}

Let $\spec{\rho}$ denote the vector of non-zero eigenvalues of $\rho$ in non-increasing order.

\begin{theorem}
  \label{thm:spectrumcondition}
  Let $\rho_{AB}$ be a state that allows a pure symmetric extension to $\proj{\psi}_{ABB^\prime}$. Then 
  \begin{equation}
    \label{eq:spectrumcondition}
    \spec{\rho_{AB}} = \spec{\rho_B}
  \end{equation}
\end{theorem}
\begin{proof}
  Using the Schmidt decomposition with the splitting $AB|B^\prime$ we can write the extended state as
  \begin{equation}
    \ket{\psi}_{ABB^\prime} = \sum_j \sqrt{\lambda_j} \ket{\phi_j}_{AB} \ket{j}_{B^\prime}.
  \end{equation}
  The reduced density matrices of this state are
  \begin{equation}
    \rho_{AB} = \sum_j \lambda_j \proj{\phi_j} \qquad \rho_{B^\prime} = \sum \lambda_j \proj{j}
  \end{equation}
  \ie the spectra of $\rho_{AB}$ and $\rho_{B^\prime}$ are equal. By symmetry between $B$ and $B^\prime$, $\rho_B = \rho_{B^\prime}$ so $\spec{\rho_{AB}} = \spec{\rho_{B}}$.
\end{proof}

In general, we don't expect all states that satisfy condition \eqref{eq:spectrumcondition} to have a pure symmetric extension. The following corollary provides a test that can rule out a pure-symmetric extension.
\begin{corollary}
For any state $\rho_{AB}$ that has a pure symmetric extension and any operator $M$ on $\mathcal{H}_A$, the (unnormalized) state 
\begin{equation}
  \widetilde{\rho}_{AB} = (M \otimes I_B) \rho_{AB} (M \otimes I_B)^\dagger
\end{equation}
satisfies condition \eqref{eq:spectrumcondition}.
\end{corollary}
\begin{proof}
Let $\ket{\psi}_{ABB'} = \pm \swap \ket{\psi}_{ABB'}$ be the pure symmetric extension of $\rho_{AB}$. 
The filter $M$ acts only on $\mathcal{H}_A$, so it commutes with $\swap$.
Therefore $M \ket{\psi}_{ABB'} = \pm M \swap \ket{\psi}_{ABB'} = \pm \swap M \ket{\psi}_{ABB'}$, 
so $M \ket{\psi}_{ABB'}$ is a symmetric extension of its reduced state $\widetilde{\rho}_{AB}$. 
Because of theorem \ref{thm:spectrumcondition}, $\widetilde{\rho}_{AB}$ then satisfies \eqref{eq:spectrumcondition}.
\end{proof}

This condition is useful, since if given a state 
that is not pure-extendible
but satisfies condition \eqref{eq:spectrumcondition},
applying a random filter on system $A$ will usually break the condition and reveal that it is not pure-extendible.


\section{Sufficiency for two qubits}
\label{sec:necsufqubit}

In this section it is shown that if $\rho_{AB}$ is a 2-qubit state and satisfies $\spec{\rho_{AB}} = \spec{\rho_{B}}$, then there exists a pure state $\ket{\psi}_{ABB^\prime}$ such that $\ket{\psi}_{ABB^\prime} = \swap \ket{\psi}_{ABB^\prime}$. We first start by giving an equivalent condition to the spectrum condition.
\begin{lemma}
  \label{Bsymmetric}
Given a bipartite state $\rho_{AB}$. Then $\spec{\rho_{AB}} = \spec{\rho_{B}}$ if and only if there exists a pure tripartite state $\ket{\psi}_{ABB^\prime}$ with reductions  $\rho_{AB}$, $\rho_B$ and $\rho_{B^\prime}$ where $\rho_B = \rho_{B^\prime}$.
\end{lemma}
\begin{proof}
  Assume $\spec{\rho_{AB}} = \spec{\rho_{B}} = (\lambda_j)$. We can write the states in the spectral decomposition, $\rho_{AB} = \sum_j \lambda_j \proj{\varphi_j}$, $\rho_{B} = \sum_j \lambda_j \proj{b_j}$. Then a purification of $\rho_{AB}$ is 
\begin{equation}
  \label{eq:pureABBprime}
  \ket{\psi}_{ABB^\prime} = \sum_j \sqrt{\lambda_j} \ket{\varphi_j}_{AB} \ket{b_j}_{B^\prime}.
\end{equation}
Tracing out the $AB$ system we get $\rho_{B^\prime} = \sum_j \lambda_j \proj{b_j} = \rho_{B}$.

Conversely, assume that there exists a pure (not necessarily symmetric) extension of $\rho_{AB}$, $\ket{\psi}_{ABB^\prime}$ with the reduced states $\rho_B = \rho_{B^\prime}$. In the spectral decomposition, $\rho_{B} = \rho_{B^\prime} = \sum_j \lambda_j \proj{b_j}$. A purification of $\rho_{B^\prime}$ to $ABB'$ is \eqref{eq:pureABBprime}, and the spectrum of $\rho_{AB}$ is $(\lambda_j)$, just like $\rho_B$.
\end{proof}

\begin{theorem}
\label{thm:necsufqubit}
For a two qubit state, $\spec{\rho_{AB}} = \spec{\rho_{B}}$ is a necessary and sufficient condition for it to have a pure symmetric extension.
\end{theorem}
\begin{proof}
The condition is necessary for any dimension and this is dealt with in section \ref{sec:condition}. Here we only prove sufficiency for two qubits. By lemma \ref{Bsymmetric}, the condition implies that there exists a pure state $\ket{\psi}_{ABB^\prime}$ such that $\rho_B = \rho_{B^\prime}$.
We will prove that for such a pure state, there is always a unitary operator on the $B^\prime$ system alone that will make it symmetric between $B$ and $B^\prime$. 

First, we prove the special case when $\rho_B$ is completely mixed. Then $\rho_{BB^\prime} = \tr_A \proj{\psi}_{ABB^\prime}$ is a state with maximally mixed subsystem. For such a state, there exist local unitaries $U_B, V_{B^\prime}$ such that $(U_B \otimes V_{B^\prime})\rho_{BB^\prime}(U_B \otimes V_{B^\prime})^\dagger$ is Bell-diagonal \cite{horodecki96c}. Moreover, since $A$ is a qubit, $\rho_{BB^\prime}$ is of rank-2 and we have
\begin{equation}
\label{eq:belldiagform}
(U_B \otimes V_{B^\prime})\rho_{BB^\prime}(U_B \otimes V_{B^\prime})^\dagger
= 
p \proj{\psi_1} + (1-p) \proj{\psi_2}
\end{equation}
with $\ket{\psi_1}$ and $\ket{\psi_2}$ two of the four Bell-diagonal states $\ket{\Phi^\pm} = (\ket{00} \pm \ket{11})/\sqrt{2}$, $\ket{\Psi^\pm} = (\ket{01} \pm \ket{10})/\sqrt{2}$. Since the Bell-basis can be permuted arbitrarily with local unitaries \cite{bennett96b}, we can choose $U_B$ and $V_{B^\prime}$ such that $\ket{\psi_1} = \ket{\Phi^+}$ and $\ket{\psi_2} = \ket{\Phi^-}$, so that we avoid the antisymmetric state $\ket{\Psi^-}$. The state in \eqref{eq:belldiagform} can now be purified to $\sqrt{p} \ket{0}_A \ket{\Phi^+}_{BB^\prime} + \sqrt{1-p} \ket{1}_A \ket{\Phi^-}_{BB^\prime}$. Since all purifications of a state are equivalent up to a local unitary on the purifying system --- in this case $A$ --- this is related to the pure state that we started out with as
\begin{multline}
(T_A \otimes U_B \otimes V_{B^\prime}) \ket{\psi}_{ABB^\prime}
\\
=
\sqrt{p} \ket{0}_A \otimes \ket{\Phi^+}_{BB^\prime} + \sqrt{1-p} \ket{1}_A \otimes \ket{\Phi^-}_{BB^\prime}
\end{multline}
where $T_A$ is the unitary operator on $A$ that relates this purification to the one where $A$ is left unchanged. We now perform the unitary $T_A^\dagger \otimes U_B^\dagger \otimes U_{B^\prime}^\dagger$ on the state, and a unitary of this form will not change the symmetry between $B$ and $B^\prime$. This gives
\begin{multline}
(I_A \otimes I_B \otimes U_B^\dagger V_{B^\prime}) \ket{\psi}_{ABB^\prime}
=
(T_A^\dagger \otimes U_B^\dagger \otimes U_{B^\prime}^\dagger) \times \\
(\sqrt{p} \ket{0}_A \otimes \ket{\Phi^+}_{BB^\prime} + \sqrt{1-p} \ket{1}_A \otimes \ket{\Phi^-}_{BB^\prime}).
\end{multline}
From this we can conclude that performing the unitary $U^\dagger V$ on system $B^\prime$ will take the starting state $\ket{\psi}_{ABB^\prime}$ to a symmetric one, so the state $\rho_{AB}$ has a symmetric extension.

We now consider the generic case when the reduced state $\rho_B$ is \emph{not} maximally mixed. 
In this case, the two non-degenerate eigenvectors of $\rho_B$ provide a preferred basis for $B$ and the corresponding basis in $B'$ is an eigenbasis for $\rho_{B'}$.
By choosing the bases in this way, we make sure that $\rho_{B} = \rho_{B'}$ are diagonal.

An arbitrary state vector of the system $ABB^\prime$ can written as $a \ket{000} + b\ket{001} + c\ket{010} + d\ket{011} + e\ket{100} + f\ket{101} + g\ket{110} + h\ket{111}$, where $a,\ldots,h$ are complex numbers whose absolute square sum to 1. 
It is symmetric under permutation of $B$ and $B^\prime$ iff $b = c$ and $f = g$. 
In appendix \ref{app:sufcalculation} we show that imposing that the reduced states $\rho_B$ and $\rho_{B'}$ are equal, diagonal and are not maximally mixed, 
implies that the amplitudes satisfy 
\begin{equation}
\label{eq:absrelationmain}
  |b| = |c|, \quad |f| = |g|
\end{equation}
and
\begin{equation}
\label{eq:phaserelationmain}
  |c||g| \left(\e^{\im(\phi_b - \phi_c)} - \e^{\im(\phi_f - \phi_g)} \right) = 0.
\end{equation}
where $b = |b| \e^{\im \phi_b}$ and similarly for $c$, $f$ and $g$.
So while the absoulute values of the relevant amplitudes are equal, the complex phases might be off. 
This can be corrected with a phase gate on $B'$ as follows. 
If $b = c = 0$, the unitary operator on $B'$ is 
\begin{equation}
  U_{B^\prime} = \proj{0} + \e^{-\im(\phi_f - \phi_g)} \proj{1}
\end{equation}
and if $f = g = 0$ it is
\begin{equation}
  U_{B^\prime} = \proj{0} + \e^{-\im(\phi_b - \phi_c)} \proj{1}.
\end{equation}
If none of the relavant amplitudes are zero, \eqref{eq:phaserelationmain} implies that the two expressions are equal, so the same unitary operator will correct both amplitude relations.

Hence, for two-qubit states $\rho_{AB}$ that satisfy the spectrum condition \eqref{eq:spectrumcondition}, we have shown that there exists a pure state vector $\ket{\psi}_{ABB^\prime}$ which is symmetric, $\ket{\psi}_{ABB^\prime}$ = $P_{BB^\prime} \ket{\psi}_{ABB^\prime}$.
\end{proof}

This theorem, together with corollary \ref{cor:pure-extendible-decomp}, fully characterizes the set of 2-qubit states with symmetric extension. 
It is the convex hull of the set of states that satisfies condition \eqref{eq:spectrumcondition}. 
Not all the states that satisfy \eqref{eq:spectrumcondition} are extremal, however. 
While any pure-extendbile states that is itself pure (i.e.~a product state) is extremal for both the set of states and the subset of extendible states, 
there are some mixed pure-extendible states that are not extremal.
The following proposition characterizes the mixed non-extremal pure-extendible states of 2 qubits.

\begin{proposition}
\label{prop:extremal}
For a 2-qubit mixed pure-extendible state $\rho_{AB}$ the following are equivalent: \\
1. $\rho_{AB}$ can be written as a convex combination of other pure-extendible states \\
2. $\rho_{AB}$ is separable \\
3. $\rho_{AB}$ is of the form 
\begin{equation}
\label{eq:purextsepform}
\rho_{AB} = \lambda \proj{\psi_0 0} + (1-\lambda) \proj{\psi_1 1},
\end{equation}
where $\braket{0}{1} = 0$, $\braket{\psi_0}{\psi_1}$ is arbitrary and $0 < \lambda < 1$.
\end{proposition}
\begin{proof}
3 $\Rightarrow$ 2 is trivial as \eqref{eq:purextsepform} is a convex combination of two product states. 2 $\Rightarrow$ 1 is also trivial, since any mixed separable state can be decomposed into a convex combination of pure product states $\rho_{AB} = \sum_j p_j \proj{\psi_j \phi_j}$ and the product states have the pure symmetric extension $\ket{\psi_j}_A\ket{\phi_j}_B\ket{\phi_j}_{B^\prime}$.

The only nontrivial part is 1 $\Rightarrow$ 3. For this part assume that $\rho_{AB}$ can be written as a convex combination of other pure-extendible states,
\begin{equation}
  \label{eq:pureABdecomp}
  \rho_{AB} = \sum_j p_j \rho_{AB}^j
\end{equation}
where $\rho_{AB}$ and all $\rho_{AB}^j$ satisfy the spectrum condition \eqref{eq:spectrumcondition}. Tracing out $A$ gives 
\begin{equation}
  \label{eq:pureBdecomp}
  \rho_{B} = \sum_j p_j \rho_{B}^j.
\end{equation}
Since $\rho_{AB}$ has support on a 2-dimensional subspace, the support of the $\rho_{AB}^j$ must be on that same subspace. We can parametrize the states on $AB$ by Pauli operators $I_S, \Sigma_x, \Sigma_y, \Sigma_z$ on this 2-dimensional subspace, 
\begin{equation}
  \label{eq:paulidecompAB}
  \rho_{AB}^j = \frac{1}{2} (I_S + X_j \Sigma_x + Y_j \Sigma_y + Z_j \Sigma_z) = \frac{1}{2} ( I + \vec{R} \cdot \vec{\Sigma}).
\end{equation}
Note that the $I_S$ here is not the identity on the 4-dimensional Hilbert space of the system $AB$, but a projector to the 2-dimensional support of $\rho_{AB}$. The reduced states on system $B$ can be written as
\begin{equation}
   \label{eq:paulidecompB}
  \rho_{B}^j = \frac{1}{2} (I_B + x_j \sigma_x + y_j \sigma_y + z_j \sigma_z) = \frac{1}{2} ( I_B + \vec{r} \cdot \vec{\sigma}),
\end{equation}
where $\sigma_x, \sigma_y, \sigma_z$ are the Pauli operators on the qubit $B$. 
Similarly, we can write $\rho_{AB}^j = \vec{R}_j \cdot \vec{\Sigma}$ and $\rho_{B}^j = \vec{r}_j \cdot \vec{\sigma}$. 
In this representation, \eqref{eq:pureABdecomp} and \eqref{eq:pureBdecomp} becomes $\vec{R} = \sum_j p_j \vec{R}_j$ and $\vec{r} = \sum_j p_j \vec{r}_j$. 

The eigenvalues of $\rho_{AB}$ and $\rho_B$ are determined by the length of the vectors
\begin{align}
 \label{eq:r_ev_AB}
  \spec{\rho_{AB}} &= \frac{1}{2}(1 + |\vec{R}|,1 - |\vec{R}|) \\
 \label{eq:r_ev_B}
  \spec{\rho_B}    &= \frac{1}{2}(1 + |\vec{r}|,1 - |\vec{r}|)
\end{align}
The $\rho_{AB}^j$ and $\rho_{AB}$ are pure-extendible, so they satisfy \eqref{eq:spectrumcondition}. In terms of the above parametrization, this means that $|\vec{R}_j| = |\vec{r}_j|$ and $|\vec{R}| = |\vec{r}|$. 

Since tracing out a part of a quantum system never can increase the trace distance between the states \cite{ruskai94a}, we have 
\begin{equation}
\| \rho_{AB}^j - \rho_{AB}^k\|_1 \geq \| \rho_{B}^j - \rho_{B}^k\|_1
\end{equation}
The trace distance can be written in terms of $\vec{R}_j$ and $\vec{r}_j$ as $\| \rho_{AB}^j - \rho_{AB}^k\|_1 = |\vec{R}_j - \vec{R}_k|$ and  $\| \rho_{B}^j - \rho_{B}^k\|_1 = |\vec{r}_j - \vec{r}_k|$. 
From $|\vec{R}_j - \vec{R}_k|^2 \geq |\vec{r}_j - \vec{r}_k|^2$ we get
\begin{equation*}
   |\vec{R}_j|^2 - 2\vec{R}_j \cdot \vec{R}_k + |\vec{R}_k|^2 \geq |\vec{r}_j|^2 - 2\vec{r}_j \cdot \vec{r}_k + |\vec{r}_k|^2
\end{equation*}
and since $|\vec{R}_j| = |\vec{r}_j|$, this gives
\begin{equation}
  \label{eq:ip_ineq_sepproof}
  \vec{R}_j \cdot \vec{R}_k \leq \vec{r}_j \cdot \vec{r}_k.
\end{equation}

Now we can use $|\vec{R}| = |\vec{r}|$ and \eqref{eq:ip_ineq_sepproof} to show that when $\rho_{AB}$ is a pure-extendible state, the trace distance between the $\rho_{AB}^j$ does not decrease when system $A$ is traced out.
\begin{align}
 |\vec{R}|^2 &= \biggl( \sum_j p_j \vec{R}_j \biggr) \biggl( \sum_k p_k \vec{R}_k \biggr) \\
             &= \sum_j p_j^2 |\vec{R}_j|^2 + 2 \sum_{j < k} p_j p_k \vec{R}_j \cdot \vec{R}_k \\
 |\vec{r}|^2 &= \sum_j p_j^2 |\vec{r}_j|^2 + 2 \sum_{j < k} p_j p_k \vec{r}_j \cdot \vec{r}_k 
\end{align}
so by demanding $|\vec{R}| = |\vec{r}|$ and using $|\vec{R}_j| = |\vec{r}_j|$ we get
\begin{equation}
 \sum_{j < k} p_j p_k \vec{R}_j \cdot \vec{R}_k = \sum_{j < k} p_j p_k \vec{r}_j \cdot \vec{r}_k
\end{equation}
By \eqref{eq:ip_ineq_sepproof} none of the terms on the LHS can be greater than the corresponding term on the RHS. The only way for this to be satisfied is that 
\begin{equation}
 \vec{R}_j \cdot \vec{R}_k = \vec{r}_j \cdot \vec{r}_k
\end{equation}
for all pairs $(j,k)$.
By reversing the calculation leading to \eqref{eq:ip_ineq_sepproof} we get that $|\vec{R}_j - \vec{R}_k|^2 = |\vec{r}_j - \vec{r}_k|^2$ and 
\begin{equation}
\label{eq:equal_trnorm}
\| \rho_{AB}^j - \rho_{AB}^k\|_1 = \| \rho_{B}^j - \rho_{B}^k\|_1.
\end{equation}

The next step is to use \eqref{eq:equal_trnorm} to find the structure of the support of $\rho_{AB}$.
The difference $\rho_{AB}^j - \rho_{AB}^k$ must be on the same two-dimensional subspace that all the $\rho_{AB}^j$ are confined to. Being the difference between two operators with trace one, it is also traceless, so in the spectral decomposition it can be written as
\begin{equation}
\label{eq:orth_decomp}
 \rho_{AB}^j - \rho_{AB}^k = r \proj{\psi_+} - r \proj{\psi_-}
\end{equation}
for some $r \geq 0$. 
The orthogonal vectors $\ket{\psi_+}$ and $\ket{\psi_-}$ define the two-dimensional support of $\rho_{AB}^j$ and $\rho_{AB}$. 
From \eqref{eq:equal_trnorm} and taking the trace norm of both sides of \eqref{eq:orth_decomp} it is clear that $\|  \rho_{B}^j - \rho_{B}^k \|_1 = 2r$. 

Let $\rho_B^+ = \tr_A \proj{\psi_+}$ and $\rho_B^- = \tr_A \proj{\psi_-}$. 
Tracing out the $A$ system in \eqref{eq:orth_decomp} and taking the trace norm gives 
$r \| \rho_B^+ - \rho_B^- \|_1 = \|  \rho_{B}^j - \rho_{B}^k \|_1 = 2r$, or 
\begin{equation}
 \| \rho_B^+ - \rho_B^- \|_1 = 2.
\end{equation}
This is the maximal distance between two states in trace norm, and it means that $\rho_B^+$ and $\rho_B^-$ have support on orthogonal subspaces. Since $B$ is a qubit, $\rho_B^+$ and $\rho_B^-$ must be orthogonal pure states which we denote $\rho_B^+ = \proj{0}$, $\rho_B^- = \proj{1}$. This also means that 
$\ket{\psi_+}$ and $\ket{\psi_-}$ are product states,
\begin{align}
 \ket{\psi_+} &= \ket{\psi_0} \otimes \ket{0} \\
 \ket{\psi_-} &= \ket{\psi_1} \otimes \ket{1},
\end{align}
where $\ket{\psi_0}$ and $\ket{\psi_1}$ are arbitrary.

Any state on the subspace spanned by $\ket{\psi_+}$ and $\ket{\psi_-}$ can be expressed as 
\begin{equation}
 \rho_{AB} = \sum_{m,n = 0}^1 \rho_{mn} \ket{\psi_m}\bra{\psi_n} \otimes \ket{m}\bra{n}
\end{equation}
with the reduced state being
\begin{equation}
 \rho_{B} = \sum_{m,n = 0}^1 \rho_{mn} \braket{\psi_n}{\psi_m} \otimes \ket{m}\bra{n}.
\end{equation}
Since $\rho_{AB}$ is pure-extendible, it satisfies \eqref{eq:spectrumcondition} and for qubits this is equivalent to the condition that the purities of the global and reduced states are equal, $\tr(\rho_{AB}^2) = \tr(\rho_{B}^2)$. The purities are
\begin{align}
 \tr(\rho_{AB}^2) &= \rho_{00}^2 + \rho_{11}^2 + |\rho_{01}|^2 \\
 \tr(\rho_{B}^2) &= \rho_{00}^2 + \rho_{11}^2 + |\rho_{01}|^2 |\braket{\psi_0}{\psi_1}|^2.
\end{align}
For the purities to be equal, either $\rho_{01}=0$ or $|\braket{\psi_0}{\psi_1}| = 1$. In the first case, the state would be 
\begin{equation}
 \label{eq:sepform2}
 \rho_{AB} = \rho_{00} \proj{\psi_0 0} + \rho_{11} \proj{\psi_1 1}
\end{equation}
which is the sought separable form. In the other case $\ket{\psi_0}$ and $\ket{\psi_1}$ only differs by a phase, so all states in the subspace are product states of the form $\proj{\psi_0} \otimes \rho_B$ which is the special case of \eqref{eq:sepform2} where $\ket{\psi_0} = \ket{\psi_1}$.
\end{proof}


\section{Counterexamples for systems with higher dimension}
\label{sec:counterex}

In the previous section we have seen that the spectrum condition \eqref{eq:spectrumcondition} is not only necessary but also sufficient for the state to have a pure symmetric extension when the system considered is a pair of qubits. One might ask if the same might be true for any higher dimensional system. We show some counterexamples that exclude this possibility for any dimension greater than $2 \times 2$.

\begin{example}
\label{ex:4x2}
\begin{figure}
 \resizebox{\columnwidth}{!}{\includegraphics{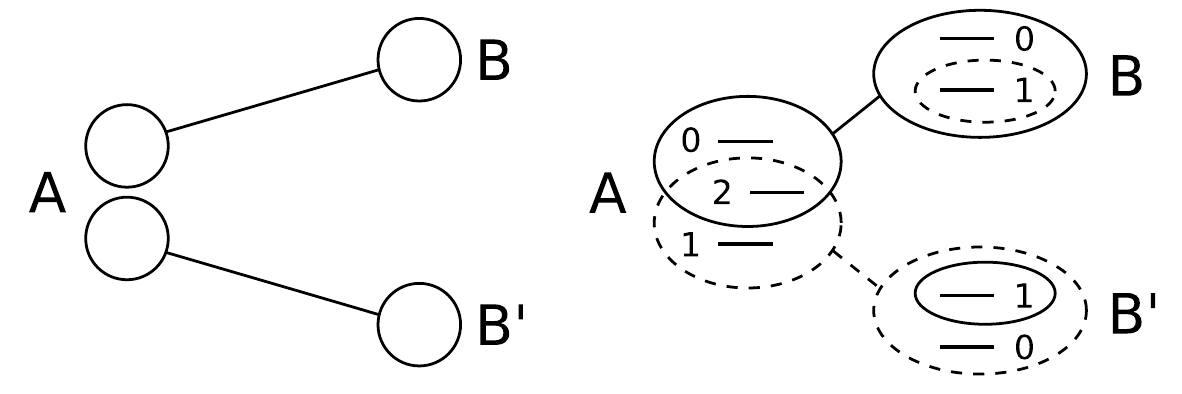}}
  \caption{\label{fig:counterexNx2} Examples of tripartite states states where $\rho_{AB}$ satisfies the spectrum condition \eqref{eq:spectrumcondition}, but does not have a symmetric extension. The $4 \times 2$ state from example \ref{ex:4x2} on the left and the $3\times 2$ state from example \ref{ex:3x2} on the right.}
\end{figure}
($4 \times 2$). The simplest example is when Alice holds two qubits and Bob one. One of Alice's qubits is maximally mixed, while the other is maximally entangled with Bob's qubit.
\begin{equation}
  \rho_{A_1 A_2 B} = \frac{I_{A_1}}{2} \otimes \proj{\Phi^+}_{A_2 B}
\end{equation}
The global density matrix $\rho_{A_1 A_2 B}$ has non-zero eigenvalues $\{1/2,1/2\}$, and so has the local one $\rho_{B}$. The state therefore satisfies the spectrum condition, but does not have a symmetric extension, since by tracing out $A_1$, Alice can make a pure maximally entangled state.
The purification of the state is illustrated in figure \ref{fig:counterexNx2}.
\end{example}

While the above example is conceptually simple, it does not exclude that the spectrum condition could be sufficient when Alice holds a qutrit. The following example is similar in spirit to the above, and shows that for system of size $3 \times 2$ and higher, the spectrum condition cannot be sufficient. 
\begin{example}
\label{ex:3x2}
($3 \times 2$). Consider the (unnormalized) vectors of a tripartite system 
\begin{align}
  \ket{v_1} &= \ket{001} + \ket{211} \\
  \ket{v_2} &= \ket{110} + \ket{211}, 
\end{align}
where the registers are $A$, $B$ and $B^\prime$. The vectors are illustrated in figure \ref{fig:counterexNx2}, the solid line corresponds to $\ket{v_1}$ and the dashed line to $\ket{v_2}$. The vector $\ket{v_1}$ is entangled between $A$ and $B$, while $\ket{v_2}$ is entangled between $A$ and $B^\prime$. Interchanging 0 and 1 at $A$ and swapping $B$ and $B^\prime$, takes $\ket{v_1}$ to $\ket{v_2}$ and \emph{vice versa}. 
Adding the two vectors and normalizing gives the state
\begin{equation}
  \ket{\psi} = \frac{1}{\sqrt{6}} \ket{001} + \frac{1}{\sqrt{6}} \ket{110} + \sqrt{\frac{2}{3}} \ket{211}
\end{equation}
The reduced states are 
\begin{align}
 \rho_{AB} &= \frac{5}{6} \proj{\psi_{1/5}} + \frac{1}{6} \proj{11} \\
 \rho_B &= \frac{5}{6} \proj{1} + \frac{1}{6} \proj{0},
\end{align}
where 
\begin{equation}
  \ket{\psi_{1/5}} = \frac{1}{\sqrt{5}} \ket{00} + \sqrt{\frac{4}{5}} \ket{21}.
\end{equation}
The non-zero eigenvalues are the same for $\rho_{AB}$ and $\rho_{B}$, so $\rho_{AB}$ satisfies the spectrum condition. However, it does not have a symmetric extension. This is most easily seen by applying the filter $F = \proj{0} + \proj{2}$ to $A$. This succeeds with probability 5/6 and the state after the filter is the pure entangled state $\ket{\psi_{1/5}}$, which has no symmetric extension.
\end{example}

Both examples above are states that can be extended to states that are invariant under some $U_A \otimes \swap$, where $U_A$ is a unitary on $A$, but not under $I_A \otimes \swap$. For the $4 \times 2$ case, $U_A$ was the unitary swapping $A_1$ and $A_2$, while in the $3 \times 2$ example it was ${\ket{0}\bra{1}} + {\ket{1}\bra{0}} + {\ket{2}\bra{2}}$. One can use the same arguments as in the proof of theorem \ref{thm:spectrumcondition} to show that any pure state that has a symmetry of the type $U_A \otimes \swap$ has a reduction to $AB$ that satisfies the condition \eqref{eq:spectrumcondition}.

The above examples show that the condition \eqref{eq:spectrumcondition} cannot be sufficient for pure extendibility for $M \times N$ systems where $M \ge 3$ and $N \ge 2$. This leaves open the question whether it is sufficient for $2\times N$ for any $N > 2$. We therefore now give an example of a class of states with system dimension $2 \times 3$ that satisfies condition \eqref{eq:spectrumcondition}, but has no symmetric extension. 

\begin{example}
($2 \times 3$). Consider states with spectral decomposition
\begin{equation}
  \rho_{AB} = \sum_{j=0}^2 \lambda_j \proj{\psi_j}
\end{equation}
where the eigenvectors are $\ket{\psi_0} = \ket{12}$, $\ket{\psi_1} = \ket{02}$ and $\ket{\psi_2} = \sqrt{s}\ket{00} + \sqrt{1-s}\ket{11}$. For such a state to satisfy the spectrum condition \eqref{eq:spectrumcondition}, the eigenvalues must be $\lambda_0 = s/2$, $\lambda_1 = (1-s)/2$ and $\lambda_2 = 1/2$. To $\rho_{AB}$ we now apply a filter operation in the standard basis in the $A$ system, $F = \sqrt{p} \proj{0} + \proj{1}$. This is a 1-SLOCC operation and cannot break a symmetric extension. After a successful filter the global and local eigenvalues $\lambda_j(p)$ and $\lambda_j^B(p)$ are 
\begin{equation*}
  \begin{aligned}
    \lambda_0(p) &= \frac{s}{1+p}\\
    \lambda_1(p) &= \frac{(1-s)p}{1+p}\\
    \lambda_2(p) &= \frac{1-s(1-p)}{1+p}
  \end{aligned}
\qquad
  \begin{aligned}
    \lambda_0^B(p) &= \frac{sp}{1+p}\\
    \lambda_1^B(p) &= \frac{1-s}{1+p}\\
    \lambda_2^B(p) &= \frac{1-(1-s)(1-p)}{1+p}.
  \end{aligned}
\end{equation*}
Except when $s \in \{0,1/2,1\}$ or $p \in \{0,1\}$, the spectra of the local and global density matrices are different. Since a filtering like this will keep a pure symmetric extension if the original state had one, $\rho_{AB}$ cannot have a pure symmetric extension. For $1/2 < s < 1$ and $0 < p < 1$ the state has no symmetric extension at all. This is because in this regime the coherent information $I(A\rangle B) := S(\rho_{B}) - S(\rho_{AB})$, where $S(\cdot)$ is the von Neumann entropy, is positive. And this is a lower bound to the distillable entanglement with one-way communication from $A$ to $B$ \cite{devetak05a}. By monogamy of entanglement, $\rho_{AB}$ cannot have a symmetric extension.
\end{example}


\section{Symmetric extension of 2-qubit states}
\label{sec:2qubits}

In previous sections we have characterized the \emph{extremal} symmetric extendible two-qubit states as those states that satisfy \eqref{eq:spectrumcondition} but not \eqref{eq:purextsepform}.
We would also like to extend this to a characterization of \emph{all} states with symmetric extension. 
In other words we want necessary and sufficient conditions for the ability to write a state as a convex combination of states that satisfy 
$\spec{\rho_{AB}^i} = \spec{\rho_{B}^i}$.
This is similar to the separability question, 
where the extremal states are pure product states, 
which are characterized by the more restrictive condition 
$\spec{\rho_{AB}^i} = \spec{\rho_{B}^i} = (1,0,\ldots)$.
Many years of entanglement theory have taught us that even though product states
are easy to recognize, the separable states are not, 
except in special cases (two qubits is one of them).
For one thing, even though the pure product states can be characterized through its local 
and global spectrum, 
we need to know more about the structure to decide if a state is separable --- 
even for two qubits \cite{nielsen01a}.
Nevertheless, we conjecture that two-qubit symmetric extendible states can be characterized solely by the local and global eigenvalues. 
We present a conjectured necessary and sufficient condition which is supported by numerical evidence and we can prove in some special cases.

\begin{conjecture}
\label{conj:twoqubitextension}
A two qubit state $\rho_{AB}$ with reduced state $\rho_{B}$ has a symmetric extension if and only if
\begin{equation}
\label{eq:twoqubitextension}
  \tr(\rho_B^2) \geq \tr(\rho_{AB}^2) - 4\sqrt{\det(\rho_{AB})}
\end{equation}
\end{conjecture}

Using techniques from previous sections, 
we prove the conjecture for states of rank 2.
For Bell-diagonal states necessary and sufficient conditions have been derived using 
techniques from semidefinite programming \cite{myhr08suba}, 
and we show that our conjecture is equivalent to these conditions.
Finally we show that the conjecture is also true for another special class of states.


\subsection{Rank-2 states}

When $\rho_{AB}$ has rank 2 the determinant in \eqref{eq:twoqubitextension} vanishes, 
and since the remaining inequality only compares the purity of the states, 
we can as well use the maximum eigenvalues to compare it.

\begin{theorem}
\label{thm:rank2}
A 2-qubit state $\rho_{AB}$ of rank 2 has a symmetric extension if and only if 
\begin{equation}
\label{eq:rank2cond}
\lambda_\text{max}(\rho_{AB}) \leq \lambda_\text{max}(\rho_B)
\end{equation}
\end{theorem}
\begin{figure}
\includegraphics{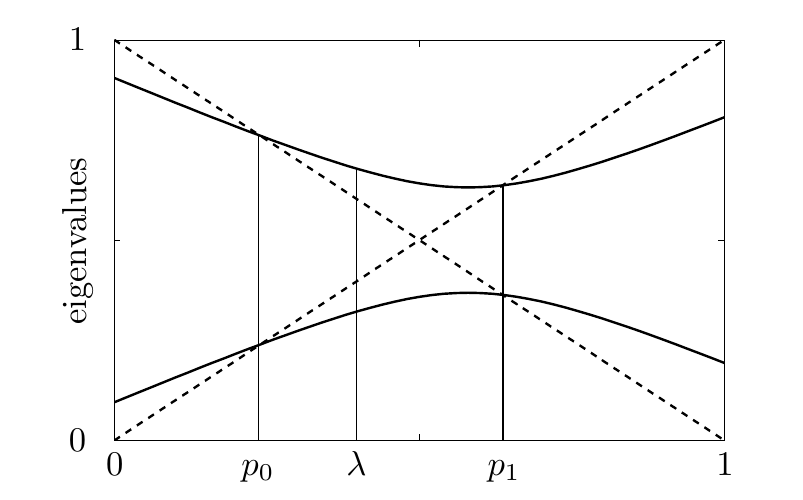}
\caption{\label{fig:rank2} Decomposition into pure-extendible states for two qubit states of rank 2. The dashed lines are global eigenvalues as parametrized on the $x$-axis. The solid lines are the local eigenvalues. The state with global eigenvalues $(1-\lambda,\lambda)$ is a convex combination of states with global eigenvalues $(1-p_0,p_0)$ and $(p_1,1-p_1)$, which have the same local eigenvalues and therefore a pure symmetric extension.} 
\end{figure}
\begin{proof}
We first prove the ``if'' part. Assume that $\rho_{AB}$ is a 2-qubit state of rank 2 that satisfies \eqref{eq:rank2cond}. We can write it in the spectral decomposition
\begin{equation}
  \rho_{AB} = (1-\lambda) \proj{\psi_0} + \lambda \proj{\psi_1}.
\end{equation}
Consider the class of states with the same eigenvectors as above, parametrized by  $p$, $\rho_{AB}^p = (1-p) \proj{\psi_0} + p \proj{\psi_1}$. 
Now $\rho_{AB} = \rho_{AB}^\lambda$. 
For $p=0$ and $p=1$ the corresponding pure states satisfy
$\lambda_\text{max}(\rho_{B}^{p}) \leq \lambda_\text{max}(\rho_{AB}^{p}) = 1$. 
Since at $p = \lambda$, $\lambda_\text{max}(\rho_{B}^{p}) \geq \lambda_\text{max}(\rho_{AB}^{p})$ by assumption and $\lambda_\text{max}$ is a continuous function of the parameter $p$, 
there must exist parameters $p_0 \in [0,\lambda]$, $p_1 \in [\lambda,1]$ such that 
$\lambda_\text{max}(\rho_{B}^{p_0}) = \lambda_\text{max}(\rho_{AB}^{p_0})$ 
and
$\lambda_\text{max}(\rho_{B}^{p_1}) = \lambda_\text{max}(\rho_{AB}^{p_1})$ (see figure \ref{fig:rank2}).  
From theorem \ref{thm:necsufqubit} we know that $\rho_{AB}^{p_0}$ and $\rho_{AB}^{p_1}$ have pure symmetric extensions, $\ket{\psi_{p_0}}_{ABB^\prime}$ and $\ket{\psi_{p_1}}_{ABB^\prime}$. Since $\rho_{AB}^\lambda$ is a convex combination $\rho_{AB}^\lambda = (1-q) \rho_{AB}^{p_0} + q \rho_{AB}^{p_1}$, where $q = (\lambda - p_0)/(p_1 - p_0)$, a symmetric extension of $\rho_{AB}^\lambda$ is $\rho_{ABB^\prime} = (1-q) \proj{\psi_{p_0}}_{ABB^\prime} + q \proj{\psi_{p_1}}_{ABB^\prime}$. 

Now for the ``only if'' part, assume that $\rho_{AB}$ is a bipartite state of rank 2 that has a symmetric extension to two copies of the qubit $B$ (in this part we do not use that $A$ is a qubit). 
Then by corollary \ref{cor:pure-extendible-decomp} it can be written as a convex combination of pure-extendible states
\begin{equation}
  \label{eq:ABdecomp}
  \rho_{AB} = \sum_j p_j \rho_{AB}^j
\end{equation}
and tracing out 
\begin{equation}
  \label{eq:Bdecomp}
  \rho_{B} = \sum_j p_j \rho_{B}^j.
\end{equation}
Like in the proof of proposition \ref{prop:extremal} we can use the fact that 
$\rho_{AB}$ has support on a 2-dimensional subspace to parametrize it using 
Pauli operators as in \eqref{eq:paulidecompAB}. 
Likewise we expand $\rho_{B}$ as in \eqref{eq:paulidecompB}, so 
\eqref{eq:ABdecomp} and \eqref{eq:Bdecomp} become
$\vec{R} = \sum_j p_j \vec{R}_j$ and $\vec{r} = \sum_j p_j \vec{r}_j$.
We can proceed exactly as in the previous proof to arrive at
\eqref{eq:ip_ineq_sepproof} which says that
$\vec{R}_j \cdot \vec{R}_k \leq \vec{r}_j \cdot \vec{r}_k$.
for all $i$ and $j$.

Since $\rho_{AB}^j$ are pure-extendible states, 
they have the same eigenvalues as the corresponding $\rho_{B}^j$ and therefore 
$|\vec{R}_j| = |\vec{r}_j|$.
Now we can use this and \eqref{eq:ip_ineq_sepproof} to compare $|\vec{R}|$ and $|\vec{r}|$,
\begin{align}
 |\vec{R}|^2 &= \biggl( \sum_j p_j \vec{R}_j \biggr) \biggl( \sum_k p_k \vec{R}_k \biggr) \\
             &= \sum_j p_j^2 |\vec{R}_j|^2 + 2 \sum_{j < k} p_j p_k \vec{R}_j \cdot \vec{R}_k \\
          &\leq \sum_j p_j^2 |\vec{r}_j|^2 + 2 \sum_{j < k} p_j p_k \vec{r}_j \cdot \vec{r}_k \\ 
             &= |\vec{r}|^2.
\end{align}
From $|\vec{R}| \leq |\vec{r}|$ and the relations to eigenvalues \eqref{eq:r_ev_AB} and \eqref{eq:r_ev_B} we can conclude that $\lambda_\text{max}(\rho_{AB}) \leq \lambda_\text{max}(\rho_B)$ which completes the proof.
\end{proof}

\begin{remark}
\label{remark:rank2necessary}
 The assumption that system $A$ is a qubit was only needed in the ``if'' part of the proof 
to conclude that states that satisfy the spectrum condition $\lambda_\text{max}(\rho_{B}) = \lambda_\text{max}(\rho_{AB})$ have a symmetric extension.
The rest of the proof, in particular the ``only if'' part, is independent of this assumption.
Therefore no $N \times 2$ state of rank 2 that satisfies $\lambda_\text{max}(\rho_{B}) < \lambda_\text{max}(\rho_{AB})$ can have a symmetric extension.
\end{remark}

\subsection{Bell-diagonal states}

Bell-diagonal states have eigenvectors 
$\ket{\Phi^\pm} = (\ket{00} + \ket{11})/\sqrt{2}$ and 
$\ket{\Phi^\pm} = (\ket{01} + \ket{10})/\sqrt{2}$, 
and are therefore defined by their eigenvalues $p_I, p_X, p_Y, p_Z$.
Any two-qubit state with maximally mixed subsystems is Bell-diagonal with 
the right choice of local basis \cite{horodecki96c}. 
For such states, necessary and sufficient conditions for symmetric extension 
have recently been found \cite{myhr08suba}. 
Parametrized by the following parameters,
\begin{subequations}
\label{eq:alpha-transformation}
\begin{align}
 \alpha_0 &:= p_I + p_X + p_Y + p_Z = 1\\
 \alpha_1 &:= p_I - p_X - p_Y + p_Z\\
 \alpha_2 &:= \sqrt{2} ( p_I - p_Z )\\
 \alpha_3 &:= \sqrt{2} ( p_X - p_Y ),
\end{align}
\end{subequations}
a state admits a symmetric extension if and only if at least one of the following inequalities is satisfied,
\begin{subequations}
\begin{gather}
\label{eq:rankoneZsufcond}
 4 \alpha_1 (\alpha_2^2-\alpha_3^2) - (\alpha_2^2 - \alpha_3^2)^2 - 4 \alpha_1^2 (\alpha_2^2 + \alpha_3^2) \geq 0 \\
\label{eq:rank2Zposconstraint1}
 \alpha_2^2-\alpha_3^2 - 2\sqrt{2}\alpha_1|\alpha_2| \geq 0 \\
\label{eq:rank2Zposconstraint2}
 \alpha_3^2-\alpha_2^2 + 2\sqrt{2}\alpha_1|\alpha_3| \geq 0.
\end{gather}
\end{subequations}

We now want to prove that these conditions are equivalent to \eqref{eq:twoqubitextension} for Bell-diagonal states. Since these states have maximally mixed subsystems, the conjectured condition becomes
\begin{equation}
  \label{eq:twoqubitextensionbell}
  4\sqrt{\det(\rho_{AB})} \geq \tr(\rho_{AB}^2) - \frac{1}{2},
\end{equation}
where $\det(\rho_{AB}) = p_I p_X p_Y p_Z$ and $\tr(\rho_{AB}^2) = p_I^2 + p_X^2 + p_Y^2 + p_Z^2$.
This is equivalent to at least one of the following inequalities holding
\begin{subequations}
\begin{align}
\label{eq:twoqubitextensionbellpurity}
 \tr(\rho_{AB}^2) & \leq \frac{1}{2}\\
\label{eq:twoqubitextensionbellsquared}
 16 \det(\rho_{AB}) &\geq [\tr(\rho_{AB}^2) - \frac{1}{2}]^2.
\end{align}
\end{subequations}
For the two sets of inequalies to be equivalent \emph{each} of the inequalities \eqref{eq:twoqubitextensionbellpurity}-\eqref{eq:twoqubitextensionbellsquared} must imply \emph{at least one} of \eqref{eq:rankoneZsufcond}-\eqref{eq:rank2Zposconstraint2} and vice versa.
By changing coordinates according to \eqref{eq:alpha-transformation} it is 
straightforward to show that \eqref{eq:twoqubitextensionbellsquared} is equivalent to \eqref{eq:rankoneZsufcond}.
For the other inequalities the relationship is more involved, but we prove that the sets of inequalities are equivalent in appendix \ref{app:belldiagproof}. 
Therefore conjecture \ref{conj:twoqubitextension} holds for Bell-diagonal states.

\subsection{Z-correlated states}

Finally, we consider states of the form
\begin{equation}
\label{eq:Zcorrform}
  \rho_{AB} = 
\begin{bmatrix}
p_1 & 0 & 0 & x \\
0 & p_2 & y & 0 \\ 
0 & y & p_3 & 0 \\ 
x & 0 & 0 & p_4 
\end{bmatrix}
\end{equation}
in the product basis $\ket{00}, \ket{01}, \ket{10}, \ket{11}$.
Without loss of generality we can assume that $p_1 \geq p_2, p_3, p_4$ and that $x$ and $y$ are both real and nonnegative, since this can be accomplished by changing the local basis.
This class includes the Bell-diagonal states as the special case 
where $p_1 = p_4$ and $p_2 = p_3$.
In this subsection we will show that the conjectured condition \eqref{eq:twoqubitextension} is necessary and sufficient in another special case of this class, namely when $y = 0$.

Let us first, however, simplify the problem for the whole class. 
The following lemma gives a necessary and sufficient condition for a state of the form \eqref{eq:Zcorrform} to have a symmetric extension.

\begin{lemma}
\label{lem:Zcorrnecsuf}
 A state of the form \eqref{eq:Zcorrform} has symmetric extension if and only if
there exist $s \in [0,p_2]$ and $t \in [0,\min(p_3,p_4)]$ such that 
\begin{subequations}
\begin{align}
  \label{eq:xcondst}
  x &\leq \sqrt{s}\sqrt{p_1-t} + \sqrt{t}\sqrt{p_4-s}\\
  \label{eq:ycondst}
  y &\leq \sqrt{s}\sqrt{p_2-t} + \sqrt{t}\sqrt{p_3-s}.
\end{align}
\end{subequations}
\end{lemma}
\begin{proof}
 For the ``if'' part we give an explicit symmetric extension of the state for the case when the inequalities are saturated. 
The extended state is then the rank-2 state 
$\rho_{ABB'} = p \proj{\psi_1} + (1-p) \proj{\psi_2}$ where
\begin{equation}
\label{eq:stextension}
\begin{aligned}
\sqrt{p}\ket{\psi_1} &= \sqrt{p_1-t}\ket{000} + \sqrt{p_2-t}\ket{011} \\
                     & \quad + \sqrt{s}\ket{101} + \sqrt{s}\ket{110}\\
\sqrt{1-p}\ket{\psi_2} &= \sqrt{t}\ket{001} + \sqrt{t}\ket{010}  \\
                     & \quad + \sqrt{p_3-s}\ket{100} + \sqrt{p_4-s}\ket{111}.
\end{aligned}
\end{equation}
If a state has symmetric extension for a given $x$ and $y$, then also states with smaller $x$ or $y$ have symmetric extension.
This is because local unitaries can change the sign of either $x$ or $y$, $I \otimes \sigma_z$ will change the sign of $x$ while $\sigma_z \otimes \sigma_z$ does the same for $y$.
The resulting states will also have a symmetric extension.
Mixing the original state with one of these states will reduce either $x$ or $y$ of the original state, and convex combinations of extendible states also have symmetric extension. 
Hence, we can have inequality instead of equality in \eqref{eq:xcondst} and \eqref{eq:ycondst}.

For the ``only if'' part, a generic symmetric operator on $ABB'$ that reduces to 
\eqref{eq:Zcorrform} when $B'$ is traced out has the form. 
\begin{equation}
 \begin{bmatrix}
  p_1 - t & \times & \times & \times & \times & k_1 & k_1 & \times \\
  \times & t      & \times & \times & l_1     & \times & \times & k_2\\
  \times & \times & t      & \times & l_1     & \times & \times & k_2\\
  \times & \times & \times & p_2 - t & \times & l_2    & l_2    & \times\\
  \times & l_1^*  & l_1^*  & \times & p_3 - s & \times & \times & \times\\
  k_1^*  & \times & \times & l_2^*  & \times  & s      & \times & \times\\
  k_1^*  & \times & \times & l_2^*  & \times  & \times & s       & \times\\
  \times & k_2^*  & k_2^*  & \times & \times  & \times & \times & p_4-s\\
 \end{bmatrix}
\end{equation}
Here $k_1 + k_2 = x$ and $l_1 + l_2 = y$. 
For this to be postive semidefinite, all subdeterminants must be positive. 
From positivity of the subdeterminants 
\begin{equation*}
 \begin{vmatrix}
  p_1 -t & k_1 \\
  k_1^*  & s
 \end{vmatrix}
 \text{ and }
 \begin{vmatrix}
  t      & k_2 \\
  k_2^*  & p_4 -s
 \end{vmatrix}
\end{equation*}
we get that $x = k_1 + k_2 \leq |k_1| + |k_2| \leq \sqrt{s}\sqrt{p_1-t} + \sqrt{t}\sqrt{p_4-s}$. 
From the subdeterminants involving $l_1$ and $l_2$ we get
$y = l_1 + l_2 \leq |l_1| + |l_2| \leq \sqrt{t}\sqrt{p_3-s} + \sqrt{s}\sqrt{p_2-t}$.
\end{proof}

Since $p_1 \geq p_2$ the possible values for $t$ in \eqref{eq:xcondst} and \eqref{eq:ycondst} are between 0 and $p_2$.
The parameter $s$, however, is bounded from above by both $p_3$ and $p_4$. 
Before we go to the special case $y = 0$ we treat the case $p_3 \geq p_4$ separately, 
since knowing which of the two bounds applies will simplify the analysis.
When $p_3 \geq p_4$ the state has a symmetric extension for any $x$ and $y$, 
since even the rank-2 state by taking the maximum 
$x = \sqrt{p_1 p_4}$ and $y = \sqrt{p_2 p_3}$ has symmetric extension 
by theorem \ref{thm:rank2}. 
It is also easy to verify that in this case $\tr(\rho_{B}^2) \geq \tr(\rho_{AB}^2)$, 
so the condition \eqref{eq:twoqubitextension} is always satisfied.

In Appendix \ref{app:Zcorrcomp} we show that when $y=0$, maximizing the bound for $x$ in 
\eqref{eq:xcondst} gives the condition
\begin{equation}
\label{eq:Zcorrboundwithyzero}
 x \leq \left\{ 
\begin{aligned}
 &\sqrt{p_1 p_4} 
    \qquad \text{for} \quad p_1 p_3 + p_2 p_4 \geq p_1 p_4 \\
 &\sqrt{p_3}\sqrt{p_1-p_2} + \sqrt{p_2}\sqrt{p_4-p_3}  
    \quad \text{otherwise}.
\end{aligned}
\right.
\end{equation}
This is also what conjecture \ref{conj:twoqubitextension} reduces to in this case.
Therefore the conjecture holds for this class of states.

Any two qubit state with three degenerate eigenvalues will be of this class. 
In this case, $\ket{00}$ and $\ket{11}$ can be taken as the Schmidt basis vectors 
of the non-degenerate eigenvector. 
We can then write the state as 
$(\lambda_1 - \lambda) \proj{\psi} + \lambda/4 I$, where $\lambda_1$ is the non-degenerate eigenvalue, $\lambda$ the degenerate eigenvalue and $\ket{\psi}$ the non-degenerate eigenvector.
Since $I$ is diagonal and $\proj{\psi}$ only has an off-diagonal entry in the $x$ position, the state is of the form for which we have just proven that our conjecture holds.


\section{Application to (anti-)degradable channels}
\label{sec:degradable}

So far we have been interested in quantum states and whether it has a symmetric extension. 
We make the connection to \emph{degradable} \cite{devetak05b} and \emph{anti-degradable} \cite{wolf07a} quantum channels which are related concepts in quantum channel theory.
If a channel is degradable or anti-degradable this greatly simplifies the evaluation of the quantum capacity of the channel.

A quantum channel can be represented by a unitary operator acting jointly on the system and the environment, where the environment starts out in a pure state, followed by tracing out the environment. Given a channel $\mathcal{N}: \mathcal{N}(\rho) = \tr_E (U(\rho \otimes \proj{0}_E)U^\dagger)$, the \emph{complementary channel} is the channel to the environment, where the system is traced out, $\mathcal{N}^C (\rho) = \tr_S (U(\rho \otimes \proj{0}_E)U^\dagger)$.
The complementary channel is only defined up to a unitary on the output system, and the channel itself
is a complementary channel of its complementary channel.
A channel $\mathcal{N}$ is called degradable if there exists another channel $\mathcal{D}$, that will degrade the channel to the complementary channel when applied on the output, $\mathcal{N}^C = \mathcal{D} \circ \mathcal{N}$. Similarly, the channel is called anti-degradable if the complementary channel is degradable, $\mathcal{N} = \mathcal{D} \circ \mathcal{N}^C$.

Using the Choi-Jamio{\l}kowski isomorphism \cite{choi75a,jamiolkowski72a} we can represent any channel by the bipartite quantum state 
resulting from the channel acting on
one half of a maximally entangled state. 
We use the convention where Alice prepares a maximally entangled state and sends the second subsystem to Bob through the channel, a procedure that leaves the first subsystem maximally mixed \footnote{A more common convention is to let the channel act on the first subsystem of the operator $M = \sum_{j,k=0}^{d-1} \ket{jj}\bra{kk}$ which is a maximally entangled state normalized such that $\tr_A(M) = I_B$.}.
\begin{equation}
 \rho_\mathcal{N} = \frac{1}{d} \sum_{j,k=0}^{d-1} \ket{i}\bra{j} \otimes \mathcal{N}(\ket{i}\bra{j})
\end{equation}

Like in the rest of this paper, we always consider symmetric extensions to two copies of the second subsystem, which in the Choi-Jamio{\l}kowski representation represents the output system.
\begin{lemma}
\label{lem:symextantidegradable}
 A channel $\mathcal{N}$ is anti-degradable if and only if its Choi-Jamo{\l}kowski representation $\rho_\mathcal{N}$ has a symmetric extension.
\end{lemma}
\begin{proof}
 Let the channel $\mathcal{N}$ be anti-degradable, and let $\mathcal{D}$ be the channel that degrades the complementary channel, $\mathcal{N} = \mathcal{D} \circ \mathcal{N}^C$.
Applying $\mathcal{N}$ on the second half of a maximally entangled state and applying $\mathcal{D}$ to the environment produces a tripartite state $\rho_{ABE}$ where the reduced states satisfy $\rho_{AB} = \rho_{AE} = \rho_\mathcal{N}$, but it does not need to be invariant under $\swapBE$. 
The state $1/2(\rho_{ABE} + \swapBE \rho_{ABE} \swapBE^\dagger)$ has the same reduced states and is also invariant under exchange of $B$ and $E$. It is therefore a symmetric extension of $\rho_\mathcal{N} = \rho_{AB}$.

Conversely, let the Choi-Jamio{\l}kowski representation $\rho_\mathcal{N}$ have a symmetric extension $\rho_{ABB^\prime}$. 
This satisfies $\rho_{AB} = \rho_{AB^\prime} = \rho_\mathcal{N}$ and has a purification $\ket{\psi}_{ABB^{\prime}R}$. 
The Choi-Jamio{\l}kowski representation of the complementary channel is then
$\rho_{AB'R}$ where $B'R$ is the output system.
Clearly, a degrading channel is then
$\mathcal{D}(\rho_{B^\prime R}) = \tr_R(\rho_{B^\prime R})$.
\end{proof}
This means that all necessary or sufficient conditions derived for symmetric extension are also necessary or sufficient conditions for the Choi-Jamio{\l}kowski representation of an anti-degradable channel. In particular, if conjecture \ref{conj:twoqubitextension} is true it will also characterize the anti-degradable qubit channels.
 
By interchanging the roles of the output and the environment, we can reduce to problem of deciding whether a channel is degradable to deciding whether the Choi-Jamio{\l}kowski representation of the complementary channel has a symmetric extension. A channel $\mathcal{N}$ with $d_A$-dimensional input, $d_B$-dimensional output and environment dimension of $d_E$ is degradable if and only if $\rho_{\mathcal{N}^C}$ of dimension $d_A \times d_E$ and rank $d_B$ has symmetric extension. Wolf and P{\'e}rez-Garc{\'i}a \cite{wolf07a} found that when $d_E = 2$, a qubit channel is either degradable, anti-degradable or both. This also follows from our theorem \ref{thm:rank2} about symmetric extension of rank-2 two-qubit states.. For qubit channels with larger environment there are examples of channels that are neither, even close to the identity channel \cite{smith07a}. Using the following theorem, we can show that no qubit channels with $d_E \geq 2$ can be degradable.

\begin{theorem}
\label{thm:limitedBrank}
 Any bipartite state $\rho_{AB}$ of rank 2 with a symmetric extension has a reduced state that satisfies $\rank(\rho_B) \leq 2$.
\end{theorem}
\begin{proof}
 By corollary \ref{cor:pure-extendible-decomp} $\rho_{AB}$ can be decomposed into pure-extendible states
\begin{equation}
 \rho_{AB} = \sum_j p_j \rho_{AB}^j
\end{equation}
where the $\rho_{AB}^j$ all satisfy the spectrum condition \eqref{eq:spectrumcondition}. Since $\rho_{AB}$ is of rank 2, 
$\rank(\rho_{AB}^j) \leq 2$ for all $j$. 

If $\max_j \rank(\rho_{AB}^j) = 1$, all the pure-extendible states are pure product states $\rho_{AB}^j = \proj{\psi_j \otimes \phi_j}$ by \eqref{eq:spectrumcondition}. Because the rank of $\rho_{AB}$ is 2, there can only be two independent product vectors, say $\ket{\phi_1 \otimes \psi_1}$ and $\ket{\phi_2 \otimes \psi_2}$, so the support of $\rho_B$ is spanned by $\psi_1$ and $\psi_2$ and therefore at most two-dimensional.

If there is at least one $j$ such that $\rank(\rho_{AB}^j) = 2$, this defines a 2-dimensional subspace where all other $\rho_{AB}^j$ must have their support. Let $\rho_{AB}^1$ be one of the $\rho_{AB}^j$ with rank 2. Let the spectral decomposition for it and its reduction to $B$ be $\rho_{AB}^1 = \gamma \proj{\phi_0} + (1-\gamma) \proj{\phi_1}$ and $\rho_B^1 = \gamma \proj{0} + (1 - \gamma) \proj{1}$, respectively, in accordance with the spectrum condition \eqref{eq:spectrumcondition}. The eigenvectors of $\rho_{AB}^1$ can be decomposed as $\ket{\phi_k} = \ket{\widetilde{\psi}_{k0}}_A \ket{0}_B + \ket{\widetilde{\psi}_{k1}}_A \ket{1}_B$, where maximum one of the four unnormalized $\ket{\widetilde{\psi}_{kl}}_A$ can be the zero vector. Since all the other $\rho_{AB}^j$ have to have support within $\mathspan\{\ket{\phi_1},\ket{\phi_2}\}$, they can only ever have reduced states $\rho_B^j$ that are supported on $\mathspan\{\ket{0},\ket{1} \}$. Therefore, also $\rho_B$ is supported on $\mathspan\{\ket{0},\ket{1} \}$ and has $\rank(\rho_{B}) \leq 2$.
\end{proof}

This reduces the $N \times M$ symmetric extension problem for states of rank 2 to $N \times 2$. 
From remark \ref{remark:rank2necessary} we already have a necessary condition for this case, 
namely that $\lambda_\text{max}(\rho_{B}) \geq \lambda_\text{max}(\rho_{AB})$. 
This also generalizes theorem \ref{thm:rank2} to give necessary and sufficient conditions for symmetric extension of a $2 \times N$ state of rank 2.
Such a state has symmetric extension if and only if  $\lambda_\text{max}(\rho_{B}) \geq \lambda_\text{max}(\rho_{AB})$ and $\rho_{B}$ is of rank 2.

From the connection between symmetric extension and anti-degradable channels in lemma \ref{lem:symextantidegradable} the following corollary automatically follows.

\begin{corollary}
 Any anti-degradable channel $\mathcal{N}$ with qubit environment has output of rank 2. If $\rho_\mathcal{N}$ is the Choi-Jamio{\l}kowski state representing the channel, $\lambda_\text{max}(\rho_\mathcal{N}) \leq \lambda_\text{max}(\tr_A \rho_\mathcal{N})$. 
\end{corollary}
Exchanging the output and the environment changes anti-degradability into degradability:
\begin{corollary}
\label{cor:degradableenvlimit}
 Any degradable channel with qubit output has $d_E \leq 2$. If $\rho_\mathcal{N}$ is the Choi-Jamio{\l}kowski state representing the channel, $\lambda_\text{max}(\tr_A \rho_\mathcal{N}) \leq \lambda_\text{max}(\rho_\mathcal{N})$.
\end{corollary}

This result has recently been independently obtained by Cubitt et~al.~\cite{cubitt08suba}.
One could imagine that theorem \ref{thm:limitedBrank} would generalize to higher rank
so that the rank of the $\rho_B$ system always would be bounded by the rank of $\rho_{AB}$ for symmetric extendible states. 
This would mean that the dimension of the environment always would be bounded by the output rank for degradable channels.
However, Cubitt et~al.~\cite{cubitt08suba} has proved that this only holds for channels with qubit and qutrit outputs. 

If the rank of a symmetric extendible state is $R$, 
the above proof can fail only if $1 < \max_j \rank(\rho_{AB}^j) < R$.
This gives the following corollary:
\begin{corollary}
\label{cor:highBrankstructure}
If $\rho_{AB}$ has a $(1,2)$-symmetric extension and $\rank(\rho_B) > \rank(\rho_{AB})$, 
then for any decomposition into pure-extendible states
\begin{equation*}
 \rho_{AB} = \sum_j p_j \rho_{AB}^j,
\end{equation*}
$\rank(\rho_{AB}^j) < \rank(\rho_{AB})$ for all $j$.
\end{corollary}
\begin{proof}
Assume that $\max_j \rank(\rho_{AB}^j) = \rank(\rho_{AB}^1) = \rank(\rho_{AB}) =: R$.
Let the spectral decomposition of $\rho_{AB}^1$ and its reduced state be
$\rho_{AB}^1 = \sum_{k=1}^R \gamma_j \proj{\phi_k}$ 
and $\rho_{B}^1 = \sum_{k=1}^R \gamma_j \proj{k}$. 
The eigenvectors of $\rho_{AB}^1$ can then be written as 
$\ket{\phi_k} = \sum_{m=1}^R \ket{\widetilde{\psi}_{km}}\ket{m}$.
Since $\rho_{AB}^1$ has the full rank of $\rho_{AB}$, 
the support of $\rho_{AB}$ must be the space spanned by the eigenvectors of $\rho_{AB}^1$.
This means that $\rho_{B}$ has support on $\mathspan\{ \{\ket{m}\}_{m=1}^R \}$ 
and therefore has rank $R$. 
Therefore, if $\rank(\rho_B) > R$ we cannot have $\max_j \rank(\rho_{AB}^j) = \rank(\rho_{AB})$.
\end{proof}

\section{Discussion and Conclusions} 

In this work we have characterized states with symmetric extension by decomposing them into states with a pure symmetric extension. 
For two-qubits we have fully characterized these pure-extendible states and quite remarkably this characterization only depends on the global and one of the local spectra of the density matrix.
Even given this result, it is rather surprising that knowledge of this information also seems to be sufficient for deciding whether or not a generic two-qubit state has a symmetric extension.
Although we cannot prove this in general, the special cases for which we prove it and extensive numerical testing suggest that our conjecture holds for all two-qubit states.
Actually, proving that the inequality \eqref{eq:twoqubitextension} describes a convex set will be sufficient for proving that it is a necessary condition for symmetric extension, since we have proven that the extremal extendible states are all contained in this set.
One way to prove the sufficiency of the condition, is to find a way to decompose any state that satisfies it either into pure-extendible states or into extendible states of any of the classes for which we have proven that the conjecture holds.

When either of the subsystems is larger than a qubit, symmetric extendibility does not only depend on local and global eigenvalues. 
In any higher dimension there are states without symmetric extension which have the same spectra as states with pure-symmetric extension. 
It would nevertheless be interesting to know if the convex hull of the states that satisfies the spectrum condition \eqref{eq:spectrumcondition}
can be characterized in a way similar to \eqref{eq:twoqubitextension}.
Such a condition would provide a useful necessary condition for a state to have a symmetric extension.

The isomorphism between quantum channels and bipartite quatum states allows us to use our results for quantum states to make some interesting statements about quantum channels.
States with symmetric extension correspond to anti-degradable channels, and by interchanging the output and the environment we can also make statements about degradable channels.
Our corollary \ref{cor:degradableenvlimit} says that if the output of a quantum channel is a qubit, it can only be degradable if the environment also is a qubit, a result that follows from our conditions on symmetric extendible states of rank 2.
When the dimension of the channel output is higher, the environment dimension of degradable channels is not always bounded by this. 
Corollary \ref{cor:highBrankstructure} gives a condition on the structure of degradable channels with higher environment dimension than output dimension.

\begin{acknowledgments}
The authors would like to thank Debbie Leung, Matthias Christandl, Andrew Doherty, Joseph Renes and Barbara Kraus and especially Marco Piani and Tobias Moroder for stimulating discussions.
Marco Piani has also provided the last part of the proof for proposition \ref{prop:extremal} and Tobias Moroder has provided some of the Matlab scripts based on YALMIP \cite{yalmip} and SDPT3 \cite{sdpt3} that were used for numerical testing.
This work was funded by the Research Council of Norway, project No.~166842/V30, the European Union through the IST Integrated Project SECOQC and the IST-FET Integrated Project QAP and by the NSERC Discovery Grant.
\end{acknowledgments}

\appendix


\section{Bosonic and fermionic extensions}
\label{app:bosonicfermionic}

In this paper we have used the term ``symmetric extension'' for extensions that are invariant under exchange of two systems, without considering if its support is on the symmetric or antisymmetric subspace or both. 
An extension that resides only on the symmetric subspace of 
$\mathcal{H}_B \otimes \mathcal{H}_{B'}$
we call a \emph{bosonic extension}, while one that resides on the anti-symmetric subspace is 
a \emph{fermionic extension}.
Generic symmetric extensions are mixtures of bosonic and fermionic extensions.
Bosonic ($+$) and fermionic ($-$) extensions satisfy $1/2(I \pm \swap) \rho_{ABB^\prime} 1/2(I \pm \swap)$ and $\swap \rho_{ABB'} = \pm \rho_{ABB'}$. 
Here we show that when the subsystem to be extended is a qubit, 
the states with symmetric and bosonic extension coincides, but this is not true in general.

\begin{proposition}
  If a quantum state $\rho_{AB}$ of dimension $N \times 2$ has symmetric extension to $\rho_{ABB^\prime}$ it also has a bosonic extension $\sigma_{ABB^\prime}$, i.~e. that satisfies also 
\begin{equation}
 \sigma_{ABB^\prime} = \frac{1}{2}(I + \swap) \sigma_{ABB^\prime} \frac{1}{2}(I + \swap)
\end{equation}
\end{proposition}
\begin{proof}
 Decompose the extended state $\rho_{ABB^\prime}$ with the spectral decomposition as in lemma \ref{lem:symspecdecomp},
\begin{equation}
 \rho_{ABB^\prime} = \sum_j \lambda_j^+ \proj{\phi_j^+} + \sum_k \lambda_k^- \proj{\phi_k^-}
\end{equation}
where $\ket{\phi_j^+} = \swap \ket{\phi_j^+}$ and $\ket{\phi_k^-} = - \swap \ket{\phi_k^-}$ are symmetric and antisymmetric, respectively. 
The vectors are of the form 
$\ket{\phi_j^\pm} = \sum_k \alpha_{jk} \ket{\psi_{jk}}_A \ket{\psi_{k}^\pm}_{BB^\prime}$ where $\ket{\psi_k^\pm}_{BB^\prime}$ are in the symmetric and antisymmetric space of $BB^\prime$. 
When $B$ and $B^\prime$ are qubits the antisymmetric space is one-dimensional and is spanned by the vector $\ket{\Psi^-} = 1/\sqrt{2}(\ket{01} - \ket{10})$. 
The antisymmetric vectors are therefore of the product form 
$\ket{\phi_k^-} = \ket{\psi_k}_A \ket{\Psi^-}_{BB^\prime}$. 
Replacing them with symmetric vectors of the form 
$\ket{\xi_k^+} = \ket{\psi_k}_A \ket{\Psi^+}_{BB^\prime}$ 
where $\ket{\Psi^+} = 1/\sqrt{2}(\ket{01} + \ket{10})$ yields a state 
\begin{equation}
 \sigma_{ABB^\prime} = \sum_j \lambda_j^+ \proj{\phi_j^+} + \sum_k \lambda_k^- \proj{\xi_k^+}
\end{equation}
which has support on the symmetric subspace. Note that $\lambda_j^+$ and $\lambda_j^-$ are no longer eigenvalues of this state. But since the reduced states of $\ket{\xi_k^+}$ are the same as for $\ket{\phi_k^-}$, we have that $\rho_{AB} := \tr_{B^\prime} \rho_{ABB^\prime} = \tr_{B^\prime} \sigma_{ABB^\prime}$, so $\sigma_{ABB^\prime}$ is a valid bosonic extension of $\rho_{AB}$. 
\end{proof}

To show that this is an effect of the low dimension of the $B$ system, we give an example of a state of two qutrits that has a fermionic but not a bosonic extension.

\begin{example}
  Consider a tripartite pure state on $ABB^\prime$ of the form
\begin{equation}
\ket{\psi} = \alpha (\ket{012} - \ket{021}) + \beta (\ket{120} - \ket{102}) + \gamma (\ket{201} - \ket{210})
\end{equation}
where $\alpha, \beta, \gamma \ne 0$. 
This is a fermionic extension of the reduced state $\rho_{AB} = \tr_{B^\prime} \proj{\psi}$. 
If $\rho_{AB}$ had a bosonic extension, a trace preserving and completely positive (TPCP) map on $B'$ would be able to convert any purification of $\rho_{AB}$ into this bosonic extension. 
If the TPCP map is given by its Kraus operators ${K_j}$ which satisfy $\sum_j K_j^\dagger K_j = I_{B'}$, the output state when applied to $\ket{\psi}$ would be
\begin{equation}
 \sigma_{ABB'} = \sum_j (I_A \otimes I_B \otimes K_j) \proj{\psi} (I_A \otimes I_B \otimes K_j)^\dagger .
\end{equation}
If $\sigma_{ABB'}$ is a bosonic extension, all the terms in this sum must be on the symmetric subspace.
Consider one of the Kraus operators, $K$.
Applying it to $\ket{\psi}$ gives 
\begin{equation}
 (I_A \otimes I_B \otimes K) \ket{\psi} = \alpha  \ket{0} \ket{\psi_0}
                                        + \beta   \ket{1} \ket{\psi_1}
                                        + \gamma  \ket{2} \ket{\psi_2}
\end{equation}
where 
$\ket{\psi_0} = \ket{1} \otimes K \ket{2} - \ket{2} \otimes K \ket{1}$, 
$\ket{\psi_1} = \ket{2} \otimes K \ket{0} - \ket{0} \otimes K \ket{2}$ and 
$\ket{\psi_2} = \ket{0} \otimes K \ket{1} - \ket{1} \otimes K \ket{0}$.
Each of the $\ket{\psi_j}$ needs to be on the symmetric subspace of 
$\mathcal{H}_B \otimes \mathcal{H}_{B'}$.
Expressing $K$ as $\sum_{jk} k_{jk}\ketbra{j}{k}$ and
imposing $\swap \ket{\psi_1} = \ket{\psi_1}$ gives us that 
$k_{01} = k_{02} = 0$ and $k_{22} = -k_{11}$.
Doing the same with the other vectors we get that $k_{jk} = 0$ for any $j \ne k$, 
$k_{00} = -k_{22}$ and $k_{11} = -k_{00}$. 
The only possible solution to this is that $K$ vanishes, 
so no nonzero $K$ applied on $B'$ can give a vector which is on the symmetric subspace.
Hence, the state $\rho_{AB}$ cannot have a bosonic extension.
\end{example}

This means that there are states $\rho_{AB}$ with a symmetric extension that cannot be extended 
to a pure state on four systems $\ket{\psi}_{ABB^\prime R}$ in such a way that $\ket{\psi}_{ABB^\prime R} = \pm \swap \ket{\psi}_{ABB^\prime R}$. 
This condition means that the extension is bosonic ($+$) or fermionic ($-$), 
but some states with symmetric extension admit neither. 
One example is if $\rho_{AB}$ does not admit a fermionic extension and 
$\sigma_{AB}$ does not admit a bosonic extension.
Then the state $1/2(\proj{0}_{A'} \otimes \rho_{AB} + \proj{1}_{A'} \otimes \sigma_{AB})$
cannot admit bosonic nor fermionic extensions.



\section{Calculations leading to (13) and (14)}
\label{app:sufcalculation}

In this appendix we show that if we on a generic 3-qubit state $a \ket{000} + b\ket{001} + c\ket{010} + d\ket{011} + e\ket{100} + f\ket{101} + g\ket{110} + h\ket{111}$ impose that its reductions $\rho_B$ and $\rho_{B'}$ are equal, diagonal and not maximally mixed, then $|b| = |c|$ and $|f| = |g|$ and 
$|c||g| (\e^{\im(\phi_b - \phi_c)} - \e^{\im(\phi_f - \phi_g)} ) = 0$.

The two reduced density matrices of this generic state are in the computational basis 
\begin{align}
  \rho_B &= 
    \begin{bmatrix}
	|a|^2 + |b|^2 + |e|^2 + |f|^2  & a c^* + b d^* + e g^* + f h^* \\
        a^* c + b^* d + e^* g + f^* h  &  |c|^2 + |d|^2 + |g|^2 + |h|^2
    \end{bmatrix}\\
  \rho_B^\prime &= 
    \begin{bmatrix}
	|a|^2 + |c|^2 + |e|^2 + |g|^2  & a b^* + c d^* + e f^* + g h^* \\
        a^* b + c^* d + e^* f + g^* h  &  |b|^2 + |d|^2 + |f|^2 + |h|^2
    \end{bmatrix}
\end{align}

The equations we get are
\begin{subequations}
\begin{align}
  |b|^2 + |f|^2 = |c|^2 + |g|^2  \label{eq:cond1}\\
   a c^* + b d^* + e g^* + f h^* &= 0 \label{eq:cond2}\\
   a b^* + c d^* + e f^* + g h^* &= 0 \label{eq:cond3}
\end{align}
\end{subequations}
where the first is from the diagonal entries of $\rho_B$ being equal to those of $\rho_{B'}$ and the others from the off-diagonal elements being 0.

Assume that $|b| \ne |c|$, then by \eqref{eq:cond1} also $|f| \ne |g|$. From \eqref{eq:cond2} and \eqref{eq:cond3} one can then isolate $e$ and $h^*$:
\begin{subequations}
\label{subeq:eh}
\begin{align}
  e   &= \frac{a(b^*f   - c^*g)   + d^*(cf   - bg)  }{|g|^2 - |f|^2}\\
  h^* &= \frac{a(c^*f^* - b^*g^*) + d^*(bf^* - cg^*)}{|g|^2 - |f|^2}
\end{align}
\end{subequations}
From this one can compute $|e|^2 - |h|^2$ and by using \eqref{eq:cond1} this simplifies to
\begin{equation}
  \label{eq:mixedcondition}
  |e|^2 - |h|^2 = |d|^2 - |a|^2 .
\end{equation}
Taken together with \eqref{eq:cond1}, this is exactly the condition that the two diagonal elements in $\rho_B$ and $\rho_{B^\prime}$ are equal, so they are completely mixed.
If the subsystems are not completely mixed, then we must have $|b| = |c|$ and $|f| = |g|$, which is \eqref{eq:absrelationmain}.

Now we want to find the relations between the complex phases of $b$, $c$, $f$ and $g$. Denote $b = |b|\e^{\im\phi_b}$, $c = |c|\e^{\im\phi_c}$, $f = |f|\e^{\im\phi_f}$ and $g = |g|\e^{\im\phi_g}$. Multiplying \eqref{eq:cond2} by $g$, \eqref{eq:cond3} by $f$, taking the difference and using $|f| = |g|$ we obtain
\begin{equation}
  a(c^*g - b^*f) + d^*(bg - cf) = 0 
\end{equation} 
Since $|c||g| = |c||f| = |b||f| = |b||g|$, this becomes
\begin{equation}
  \e^{\im \phi_g}|c||g|(a \e^{-\im \phi_b} + d^*\e^{\im \phi_c})(\e^{\im(\phi_b - \phi_c)} - \e^{\im(\phi_f - \phi_g)}) = 0.
\end{equation} 
Then at least one of the following two equations must hold. Either 
\begin{equation}
\label{eq:phaserelationapp}
  |c||g| \left(\e^{\im(\phi_b - \phi_c)} - \e^{\im(\phi_f - \phi_g)} \right) = 0,
\end{equation}
which is \eqref{eq:phaserelationmain} that we want to show, \emph{or}
\begin{equation}
\label{eq:phasealternative}
  d^*\e^{\im \phi_c} = - a \e^{-\im \phi_b}.
\end{equation}
In the case that \eqref{eq:phaserelationapp} does not hold, \eqref{eq:phasealternative} must hold, and
we will now see that this case implies that subsystem $B$ is completely mixed. 

If we insert \eqref{eq:phasealternative} into \eqref{eq:cond3} and use $|b| = |c|$ and $|f| = |g|$, we obtain
\begin{equation}
  \label{eq:phasealternativeconsequence}
  h^* g = -e f^*.
\end{equation}
Since \eqref{eq:phaserelationapp} does not hold, $|f| = |g| \ne 0$ and therefore \eqref{eq:phasealternativeconsequence} implies $|e| = |h|$. 
The condition \eqref{eq:phasealternative} already means that $|a| = |d|$, so again we have that \eqref{eq:mixedcondition} holds so the diagonal terms in $\rho_B$ are equal and we are in the maximally mixed case.

Hence, \eqref{eq:phasealternative} cannot hold since $\rho_{B}$ and $\rho_{B'}$ are not maximally mixed and therefore \eqref{eq:phaserelationapp} which is the same as \eqref{eq:phaserelationmain} must hold.

\section{Inequality relations for Bell-diagonal states}
\label{app:belldiagproof}

In this appendix we show that each of the inequalities
\eqref{eq:twoqubitextensionbellpurity}-\eqref{eq:twoqubitextensionbellsquared} implies  at least one of \eqref{eq:rankoneZsufcond}-\eqref{eq:rank2Zposconstraint2} and vice versa. 
More precisely, \eqref{eq:rankoneZsufcond} and \eqref{eq:twoqubitextensionbellsquared} 
are equivalent, 
either of \eqref{eq:rank2Zposconstraint1} and \eqref{eq:rank2Zposconstraint2} implies \eqref{eq:twoqubitextensionbellpurity} 
while \eqref{eq:twoqubitextensionbellpurity} only implies that at least one of \eqref{eq:rankoneZsufcond}-\eqref{eq:rank2Zposconstraint2} is satisfied.

We first change variables in \eqref{eq:twoqubitextensionbellpurity} and \eqref{eq:twoqubitextensionbellsquared} so that they use the same parameters as \eqref{eq:rankoneZsufcond}-\eqref{eq:rank2Zposconstraint2}. 
This gives the two inequalities 
\begin{subequations}
\begin{gather}
\label{eq:twoqubitextensionbellpurityalpha}
\alpha_1^2 + \alpha_2^2 + \alpha_3^2 \leq 1\\
\label{eq:twoqubitextensionbellsquaredalpha}
4 \alpha_1 (\alpha_2^2-\alpha_3^2) - (\alpha_2^2 - \alpha_3^2)^2 - 4 \alpha_1^2 (\alpha_2^2 + \alpha_3^2) \geq 0.
\end{gather}
\end{subequations}
\eqref{eq:twoqubitextensionbellsquaredalpha} which comes from \eqref{eq:twoqubitextensionbellsquared} the same as \eqref{eq:rankoneZsufcond}
so they are all equivalent. 

Next, we prove that either of \eqref{eq:rank2Zposconstraint1} and \eqref{eq:rank2Zposconstraint2} implies \eqref{eq:twoqubitextensionbellpurityalpha} and therefore also \eqref{eq:twoqubitextensionbellpurity}. 
Each of \eqref{eq:rank2Zposconstraint1} and \eqref{eq:rank2Zposconstraint2} can be split into two inequalities for the cases when the variable inside the absolute value is negative or nonnegative.
For each of the four inequalities 
an orthogonal change of variables allows us to express them on a standard form. 
The transformation for \eqref{eq:rank2Zposconstraint1} is 
$\alpha_1 = \sqrt{2/3} x - \sqrt{1/3} y$, 
$\alpha_2 = \pm \sqrt{1/3} x \pm \sqrt{2/3} y$ and 
$\alpha_3 = z$,
for the cases $\pm \alpha_2 \geq 0$. 
The transformations for \eqref{eq:rank2Zposconstraint2} are obtained by interchanging $\alpha_2$ with $\alpha_3$ and $\alpha_1$ with $-\alpha_1$. 
All four inequalities then become simply $x^2 + z^2 \leq 2y^2$.
The purity condition \eqref{eq:twoqubitextensionbellpurityalpha} becomes 
$x^2 + y^2 + z^2 \leq 1$
for all the transformations.
By noting that for each transformation one of the positivity conditions for the eigenvalues translates into 
$y \leq 1/\sqrt{3}$, 
we get $x^2 + y^2 + z^2 \leq 3 y^2 \leq 1$.

The last implication we need to show is that any state that satisfies
$\tr(\rho_{AB}^2) \leq 1/2$, 
or equivalently \eqref{eq:twoqubitextensionbellpurityalpha}, 
also satisfies at least one of \eqref{eq:rankoneZsufcond}--\eqref{eq:rank2Zposconstraint2}.
For this we use the proven fact that these inequalities are necessary and sufficient conditions for the state to have a symmetric extension, and therefore the set must be convex. 
Any state that satisfies \eqref{eq:twoqubitextensionbellpurity} can be written as a convex combination of states that satisfies $\tr{\rho_{AB}^2} = 1/2$, e.~g.~$(1-q_I)\proj{\Phi^+} + q_I \rho_{AB}$, $(1-q_X)\proj{\Psi^+} + q_X \rho_{AB}$, etc. for the $q_j$ that give the right purity.
Since the determinant of a state always is non-negative, 
all these extremal states satisfy \eqref{eq:twoqubitextensionbellsquared}, 
and therefore also \eqref{eq:rankoneZsufcond}, 
so they must have a symmetric extension.
Convex combinations of states with symmetric extension also have symmetric extension, 
so any state with $\tr(\rho_{AB}^2) \leq 1/2$ has symmetric extension and therefore satisfies one of 
\eqref{eq:rankoneZsufcond}--\eqref{eq:rank2Zposconstraint2}.

\section{Equivalence for Z-correlated states}
\label{app:Zcorrcomp}

In this appendix we show that for states of the class \eqref{eq:Zcorrform} with $y = 0$ and $p_4 \geq p_3$, 
the necessary and sufficient conditions for symmetric extension from lemma \ref{lem:Zcorrnecsuf} simplify to \eqref{eq:Zcorrboundwithyzero}.
Next, we show that conjecture \ref{conj:twoqubitextension} also reduces to \eqref{eq:Zcorrboundwithyzero} for this class of states.

Since $y = 0$, \eqref{eq:ycondst} is satisfied for any $s \in [0,p_3]$ and $t \in [0,p_2]$.
Our only objective is therefore to maximize the right hand side of \eqref{eq:xcondst},  
$f(s,t) := \sqrt{s}\sqrt{p_1-t} + \sqrt{t}\sqrt{p_4-s}$, 
on this domain.
Without the constraints on $s$ and $t$, this reaches its maximum value of $\sqrt{p_1 p_4}$ for any value of $(s,t)$ that satisfies
$p_1 s + p_4 t = p_1 p_4$.
Since $s \leq p_3$ and $t \leq p_2$ this maximum value may or may not be obtainable.
The maximum value of $p_1 s + p_4 t$ is $p_1 p_3 + p_2 p_4$, 
so if $p_1 p_3 + p_2 p_4 \geq p_1 p_4$, then $x = \sqrt{p_1 p_4}$ can be obtained by choosing $s = p_3$ and $t = (p_1 p_4 - p_1 p_3)/p_4 \leq p_2$.
When $p_1 p_3 + p_2 p_4 < p_1 p_4$, however, we will have 
$f(s,t) < \sqrt{p_1 p_4}$
for all possible $(s,t)$.
In this case the optimal choice of $(s,t)$ is $(p_3, p_2)$, since in the region where
$p_1 p_3 + p_2 p_4 < p_1 p_4$ the $f(s,t)$ increases both when $s$ and $t$ increases.
The maximum value for $x$ is then $\sqrt{p_3}\sqrt{p_1-p_2} + \sqrt{p_2}\sqrt{p_4-p_3}$.
Summing up, a state of the form \eqref{eq:Zcorrform} with $y = 0$ 
has a symmetric extension if and only if
\begin{equation}
\label{eq:Zcorrboundwithyzeroapp}
 x \leq \left\{ 
\begin{aligned}
 &\sqrt{p_1 p_4} 
    \qquad \text{for} \quad p_1 p_3 + p_2 p_4 \geq p_1 p_4 \\
 &\sqrt{p_3}\sqrt{p_1-p_2} + \sqrt{p_2}\sqrt{p_4-p_3}  
    \quad \text{otherwise}
\end{aligned}
\right.
\end{equation}
which is the same as \eqref{eq:Zcorrboundwithyzero}.

The remaining part is to show that the condition \eqref{eq:twoqubitextension} from conjecture \ref{conj:twoqubitextension} is equivalent to this. 
The condition is equivalent to at least one of the following two inequalities holding
\begin{subequations}
\begin{gather}
  \label{eq:twoqubitextensionsplit1}
  \tr(\rho_{AB}^2) - \tr(\rho_{B}^2) \leq 0 \\
  \label{eq:twoqubitextensionsplit2}
  4 \sqrt{\det{\rho_{AB}}} \geq |\tr(\rho_{AB}^2) - \tr(\rho_{B}^2)|.
\end{gather}
\end{subequations}
Since $y=0$, we get $\det(\rho_{AB}) = p_2 p_3(p_1 p_4 - x^2)$
and $\tr(\rho_{AB}^2) - \tr(\rho_B^2) = 2(x^2 - p_1 p_3 - p_2 p_4)$.
Inserting this into \eqref{eq:twoqubitextensionsplit1} and \eqref{eq:twoqubitextensionsplit2} and solving for $x$ gives
\begin{subequations}
\begin{align}
  \label{eq:twoqubitextensionsplit1x}
  x &\leq \sqrt{p_1 p_3 + p_2 p_4} \\
  \label{eq:twoqubitextensionsplit2x}
  x &\leq \sqrt{p_3}\sqrt{p_1-p_2} + \sqrt{p_2}\sqrt{p_4-p_3}.
\end{align}
\end{subequations}
Only one of these inequalities have to be satisfied for a state with symmetric extension, so the upper bound on $x$ is the maximum of the two.
By comparing the two bounds we find in which region each of the two is valid, and get
\begin{equation}
 x \leq \left\{ 
\begin{aligned}
 &\sqrt{p_1 p_3 + p_2 p_4} 
    \qquad \text{for} \quad p_1 p_3 + p_2 p_4 \geq p_1 p_4 \\
 &\sqrt{p_3}\sqrt{p_1-p_2} + \sqrt{p_2}\sqrt{p_4-p_3}  
    \quad \text{otherwise}
\end{aligned}
\right.
\end{equation}
The only region where $\sqrt{p_1 p_3 + p_2 p_4}$ is the valid upper bound is when it is greater than $\sqrt{p_1 p_4}$. 
Since $x$ never can exeed $\sqrt{p_1 p_4}$ for any state, this is the same as \eqref{eq:Zcorrboundwithyzeroapp}.

\bibliography{oqctbiblio,2-qubit-condition,symextbell}

\end{document}